\documentclass[aps,prl,reprint,superscriptaddress]{revtex4-1}

\usepackage[T1]{fontenc}
\usepackage{bbm}

\usepackage{titlesec}
\titleformat{\section}[block]{\bfseries}{\thesection}{5pt}{\filcenter\MakeUppercase}{}
\titlespacing*{\section}{0pt}{*4}{*2}
\titleformat{\subsection}[runin]{\bfseries}{\thesubsection}{5pt}{\scshape}{}
\titlespacing*{\subsection}{0pt}{*1}{*2}

\usepackage{graphicx}
\graphicspath{{img/}}

\usepackage{amssymb,amsmath,amsthm}

\newcommand{\paren}[1]{\left(#1\right)}
\newcommand{\brac}[1]{\left[#1\right]}
\newcommand{\curly}[1]{\left\{#1\right\}}
\newcommand{\mc}[1]{\mathcal{#1}}

\newcommand{\suml}[3]{\sum_{#1}^{#2}#3}
\newcommand{\prodl}[3]{\prod_{#1}^{#2}#3}

\newcommand{\pr}[0]{^\prime}

\newcommand{\dg}[0]{^\dagger}

\newcommand{\wtil}[1]{\widetilde{#1}}

\newcommand{\swap}[0]{{\textsc{swap}}}
\newcommand{\px}[1]{\sigma_{x}^{#1}}
\newcommand{\py}[1]{\sigma_{y}^{#1}}
\newcommand{\pz}[1]{\sigma_{z}^{#1}}

\newcommand{\ha}[0]{\ensuremath{H}}
\newcommand{\hj}[0]{\ensuremath{H_{J}}}
\newcommand{\hh}[0]{\ensuremath{H_{h}}}
\newcommand{\uti}[0]{\ensuremath{U}}
\newcommand{\utd}[0]{\ensuremath{V}}
\newcommand*{\idm}{\ensuremath{\mathbbm{1}}}
\renewcommand{\vec}[1]{\ensuremath{\boldsymbol{#1}}}

\usepackage{mathtools}
\mathtoolsset{centercolon}

\DeclareMathOperator{\diag}{diag}
\DeclareMathOperator{\hammingweight}{ham}
\DeclareMathOperator{\rt}{rt}
\DeclareMathOperator{\qrt}{qrt}

\DeclareMathOperator{\hqrt}{hqrt}
\DeclareMathOperator{\rs}{rs}
\DeclareMathOperator{\qrs}{qrs}
\DeclareMathOperator{\twoqubitgate}{tg}

\DeclareMathOperator{\Ex}{\symbf{E}}
\DeclareMathOperator{\Prob}{P}
\DeclareMathOperator{\tr}{Tr}
\DeclareMathOperator{\revfull}{R}
\DeclareMathOperator{\entanglementOp}{E}

\DeclarePairedDelimiter\floor{\lfloor}{\rfloor}
\DeclarePairedDelimiter\bra{\langle}{\vert}
\DeclarePairedDelimiter\ket{\vert}{\rangle}
\DeclarePairedDelimiter\set{\{}{\}}
\DeclarePairedDelimiter\abs{\lvert}{\rvert}
\DeclarePairedDelimiter{\norm}{\lVert}{\rVert}
\DeclarePairedDelimiterXPP\tg[1]{\twoqubitgate}{(}{)}{}{#1}
\DeclarePairedDelimiterXPP\bigo[1]{O}{(}{)}{}{#1}
\DeclarePairedDelimiterXPP\bigomega[1]{\Omega}{(}{)}{}{#1}
\DeclarePairedDelimiterXPP\diagonal[1]{\diag}{(}{)}{}{#1}
\DeclarePairedDelimiterXPP\hamming[1]{\hammingweight}{(}{)}{}{#1}
\DeclarePairedDelimiterXPP\rnumber[1]{\rt}{(}{)}{}{#1}
\DeclarePairedDelimiterXPP\qrnumber[1]{\qrt}{(}{)}{}{#1}
\DeclarePairedDelimiterXPP\hrnumber[1]{\hqrt}{(}{)}{}{#1}
\DeclarePairedDelimiterXPP\rsize[1]{\rs}{(}{)}{}{#1}
\DeclarePairedDelimiterXPP\qrsize[1]{\qrs}{(}{)}{}{#1}
\DeclarePairedDelimiterXPP\hierprod[2]{Π_{#1}}{(}{)}{}{#2}
\DeclarePairedDelimiterXPP\modular[1]{M_G}{(}{)}{}{#1}
\DeclarePairedDelimiterXPP\probability[1]{\Prob}{[}{]}{}{#1}
\DeclarePairedDelimiterXPP\expectation[1]{\Ex}{[}{]}{}{#1}
\DeclarePairedDelimiterXPP\entanglement[1]{\entanglementOp}{(}{)}{}{#1}
\DeclarePairedDelimiterXPP\rev[1]{R}{[}{]}{}{#1}

\newcommand{\parstring}[1]{P_{[#1]}}

\usepackage{booktabs}

\usepackage[colorinlistoftodos]{todonotes}

\usepackage{color}
\usepackage{xcolor}
\definecolor{dark-red}{rgb}{0.4,0.15,0.15}
\definecolor{dark-blue}{rgb}{0.15,0.15,0.4}
\definecolor{medium-blue}{rgb}{0,0,0.5}

\definecolor{mycomment}{rgb}{0.3,0.7,0.8}
\definecolor{mygray}{rgb}{0.5,0.5,0.5}
\definecolor{lightgray}{rgb}{0.95,0.95,0.95}
\definecolor{mymauve}{rgb}{0.58,0,0.82}

\usepackage{hyperref}
\hypersetup{%
	breaklinks=true,
	colorlinks,
	linkcolor={dark-blue},
	citecolor={dark-blue},
	urlcolor={medium-blue},
	pdfauthor = {Aniruddha Bapat, Eddie Schoute, Alexey V. Gorshkov and Andrew M. Childs},
	pdftitle ={Nearly optimal time-independent reversal of a spin chain},
	pdfpagemode=UseOutlines,
	pdfborder={0 0 0}
}

\usepackage[nameinlink,capitalise]{cleveref}
\crefname{figure}{Figure}{Figures}
\crefname{equation}{}{} 
\Crefname{equation}{Equation}{Equations}

\newtheorem{theorem}{Theorem}
\newtheorem{lemma}[theorem]{Lemma}

\theoremstyle{definition}

\theoremstyle{definition}
\newtheorem{protocol}[theorem]{Protocol}
\crefname{protocol}{Protocol}{Protocols}

\begin{document}


\title{Nearly optimal time-independent reversal of a spin chain}


\author{Aniruddha Bapat}
\email{ani@umd.edu}
\affiliation{Joint Center for Quantum Information and Computer Science, NIST/University of Maryland, College Park, Maryland 20742, USA}
\affiliation{Joint Quantum Institute, NIST/University of Maryland, College Park, Maryland 20742, USA}
\author{Eddie Schoute}
\email{eschoute@umd.edu}
\affiliation{Joint Center for Quantum Information and Computer Science, NIST/University of Maryland, College Park, Maryland 20742, USA}
\affiliation{Department of Computer Science, University of Maryland, College Park, Maryland 20742, USA}
\affiliation{Institute for Advanced Computer Studies, University of Maryland, College Park, Maryland 20742, USA}
\author{Alexey V.~Gorshkov}
\email{gorshkov@umd.edu}
\affiliation{Joint Center for Quantum Information and Computer Science, NIST/University of Maryland, College Park, Maryland 20742, USA}
\affiliation{Joint Quantum Institute, NIST/University of Maryland, College Park, Maryland 20742, USA}
\author{Andrew M.~Childs}
\email{amchilds@umd.edu}
\affiliation{Joint Center for Quantum Information and Computer Science, NIST/University of Maryland, College Park, Maryland 20742, USA}
\affiliation{Department of Computer Science, University of Maryland, College Park, Maryland 20742, USA}
\affiliation{Institute for Advanced Computer Studies, University of Maryland, College Park, Maryland 20742, USA}


\date{\today}

\begin{abstract}
	We propose a time-independent Hamiltonian protocol for the reversal of qubit ordering in a chain of $N$ spins.
	Our protocol has an easily implementable nearest-neighbor, transverse-field Ising model Hamiltonian
	with time-independent, non-uniform couplings.
	Under appropriate normalization,
	we implement this \emph{state reversal} three times faster than a naive approach using \swap{} gates,
	in time comparable to a protocol of \citeauthor{Raussendorf2005} [\emph{Phys.\,Rev.\,A} 72, 052301 (2005)] that requires dynamical control.
	We also prove lower bounds on state reversal by using results on the entanglement capacity of Hamiltonians and show that we are within a factor $1.502(1+1/N)$ of the shortest time possible.
	Our lower bound holds for all nearest-neighbor qubit protocols with arbitrary finite ancilla spaces
	and local operations and classical communication. Finally, we extend our protocol to an infinite family of nearest-neighbor, time-independent Hamiltonian protocols for state reversal. This includes chains with nearly uniform coupling that may be especially feasible for experimental implementation.
\end{abstract}



\maketitle


Quantum information transfer is a fundamental operation in quantum physics, and fast, accurate protocols for transferring quantum states across a physical system are likely to play a key role in the design of quantum computers and networks~\cite{divincenzo2000,kimble2008}.
For example, quantum information transfer can be used to establish long-range entanglement and is also useful for qubit routing in quantum architectures with limited connectivity~\cite{Beals2013,Childs2019}.
Extensive work has studied the implementation of various information transfer protocols, often via Hamiltonian dynamics on spin chains~\cite{Bose2007}.

Information transfer in Hamiltonian systems is governed by the spread of entanglement and has close links to Lieb-Robinson bounds~\cite{Lieb1972}, entanglement area laws~\cite{Eisert2010}, and algorithms for quantum simulation~\cite{Haah2018}.
Fundamental limits to the rate of entanglement growth are set by bounds on the \emph{asymptotic entanglement capacity}~\cite{Duer2001,Childs2003,Childs2004,Bennett2003} and more recent small incremental entangling  theorems~\cite{Bravyi2007, Acoleyen2013, Audenaert2014,Marien2016}.
We show that these limits can also be used to obtain lower bounds on the execution time of Hamiltonian protocols for information transfer.
This raises the question of whether a protocol can achieve optimality by saturating the bound.

\emph{Quantum state transfer} studies protocols for moving qubits through a spin chain~\cite{Bose2003}.
Long-range interactions can be used to speed up protocols~\cite{Gualdi2008},
but here we consider only nearest-neighbor interactions.
State transfer protocols usually assume the intermediate medium to be in a known initial state~\cite{Christandl2004,Christandl2005,Venuti2007,Banchi2011} or allow it to change in an unknown or non-trivial manner~\cite{Yao2011,Franco2008}.
Such protocols are not directly applicable when some or all spins in the chain contain data qubits that need to be transferred or maintained.

Protocols for \emph{state reversal}, also known as \emph{state mirroring}~\cite{Albanese2004}, take steps towards addressing this issue.
State reversal reverses any input state on a spin chain about the center of the chain.
Specifically, with qubit labeling $1,2,\dots,N$, state reversal corresponds to the unitary
\begin{equation}
  \label{eq:SWAPrb}
  \revfull \coloneqq \prodl{k=1}{\lfloor\frac{N}{2}\rfloor}{\swap{}_{k,N+1-k}}
\end{equation}
up to a \emph{global phase}, which is independent of the state.
State reversal is potentially a useful subroutine for the more general task of qubit routing, where we wish to apply arbitrary permutations to the qubits.
Early results in this area require the state to be in the single-excitation subspace~\cite{Shi2005}
or introduce phases in the final state that depend on a non-local property such as the number of qubits in state $\ket{1}$~\cite{Albanese2004,Karbach2005}.
These limitations were later removed by time-dependent protocols for state reversal~\cite{Raussendorf2005,Fitzsimons2006,Kumar2015}.

In this work, we propose the first \emph{time-independent} protocol for state reversal using nearest-neighbor interactions.
We show that the execution time of our protocol is nearly optimal, comparable to the time-dependent protocol given in~\cite{Raussendorf2005}.
However, as our protocol does not require dynamical control but only pre-engineered couplings, we expect it to be more experimentally feasible on near-term quantum systems.

Before presenting our state reversal protocol in more detail,
let us elaborate on the claim that it is \emph{nearly optimal}---specifically,
that it has an evolution time within a factor $1.502(1+1/N)$ of the shortest possible.
For any nearest-neighbor spin Hamiltonian $H$, a time scale follows from a normalization that limits the strength of every two-qubit interaction but allows fast local operations.
Up to local unitaries, we can write any two-qubit Hamiltonian in the \emph{canonical form}~\cite{Bennett2002}
\begin{equation}\label{eq:CanonicalForm}
	K \coloneqq \sum_{j \in \set{x,y,z}} \mu_j \sigma_j \otimes \sigma_j \,,
\end{equation}
where $\mu_x \geq \mu_y \geq \abs{\mu_z} \geq 0$ and $\sigma_j$ are the Pauli matrices.
We impose the normalization condition that $\norm{K} = \sum_j \abs{\mu_j} \leq 1$ for all interactions,
where $\norm{\cdot}$ is the spectral norm.
Under this normalization, a \swap~can be optimally implemented in time $3\pi/4$~\cite{Vidal2002},
and our protocol achieves state reversal in time
\begin{equation}
	t_N \coloneqq \pi\sqrt{(N+1)^2 - p(N)}/4\,,\label{eq:reversalTime}
\end{equation}
where $p(N) \coloneqq N \pmod{2}$.
This is equivalent in time to a \swap{} gate circuit of depth ${\sim} N/3$.
As state reversal using only \swap{}s requires depth at least $N-1$~\cite{Alon1994},
our protocol is faster than any \swap{}-based protocol by an asymptotic factor of 3.
Similarly, we can compare to other time-independent Hamiltonian protocols that use nearest-neighbor interactions:
\cite{Christandl2004} implements state transfer in time $N\pi/4$
and \cite{Albanese2004} implements state reversal in time $N\pi/2$
but introduces relative phases in the state as mentioned earlier.
Our time-independent protocol (and some time-dependent protocols~\cite{Raussendorf2005,Fitzsimons2006,Kumar2015})
thus improve upon these previous protocols for state transfer
and state reversal except for a subleading term.

We lower-bound the time for state reversal, which can generate entanglement across a bipartition,
by using bounds on the asymptotic entanglement capacity in a more general model~\cite{Bennett2003,Childs2003}.
The asymptotic entanglement capacity bounds the rate at which entanglement can be generated
by any evolution of a given bipartite Hamiltonian interspersed with arbitrary local operations and classical communication (LOCC) and with arbitrary finite local ancilla spaces.
We give an explicit example of entanglement generated by state reversal
and lower-bound the time using the capacity of a normalized two-qubit interaction in canonical form~\eqref{eq:CanonicalForm},
even allowing for LOCC.
Nonetheless, our state reversal protocol is able to nearly saturate this bound
without classical communication, without ancillas,
and with only nearest-neighbor interactions throughout the chain.

\begin{figure}[tb]
	\centering
	\includegraphics[angle=-90, width=\linewidth]{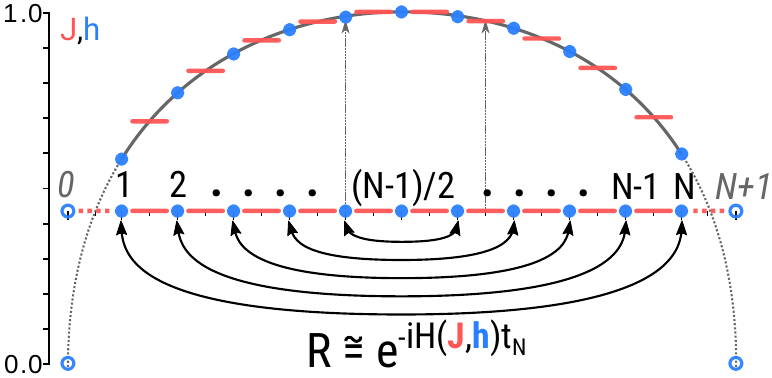}
	\caption{%
		The state reversal operation $\revfull$ (depicted by arrows)
		and an illustration of our time-independent protocol to implement it.
		The nearest-neighbor $\px{k}\px{k+1}$ couplings ($J_k$, red) and on-site $\pz{k}$ fields ($h_k$, blue) are plotted on the $y$-axis.
		Sites $0,N+1$ are ancilla qubits, which are not part of the protocol and are used purely in the analysis.
	}\label{fig:ti}
\end{figure}

We propose a state reversal protocol with Hamiltonian of the form
\begin{equation}
  \label{eq:Hform}
  H(\vec{J},\vec{h}) = J_0\px{1} + \suml{k=1}{N-1}{J_{k}\px{k}\px{k+1}} + J_{N}\px{N} - \suml{k=1}{N}{h_k\pz{k}},
\end{equation}
where the coefficients $\vec{J},\vec{h}$ are engineered as follows.
Letting
\begin{equation}
a_k \coloneqq \pi\sqrt{(N+1)^2 - (N+1-k)^2}/(4t_N)\,,
\end{equation}
for $k\in \mathbb N$,
our protocol is defined as (see also~\cref{fig:ti})
\begin{protocol}\label{prot:ti}
	Let $J_k= a_{2k+1}, h_{k}=a_{2k}$ for all sites $k$, and let $\ha:=H(\vec J,\vec h)$.
	Apply $\uti:= e^{-it_N \ha}$ to the input state.
\end{protocol}

We show in the following sections that our protocol implements state reversal exactly,
up to a global phase (we denote this equivalence by $\cong$).
In other words,
\begin{theorem}
  \label{thm:main}
$\uti \cong \revfull$.
\end{theorem}

\subsection{Proof and analysis of the protocol.}
We prove the correctness of our protocol (i.e., \cref{thm:main}) by mapping the spin chain to a doubled chain of Majorana fermions via a Jordan-Wigner transformation,
describing the action in the Majorana picture, and then mapping back to the spin picture.
To help with the analysis,
we extend the chain with two ancillary sites $\curly{0,N+1}$ called the \emph{edge}, $E$,
and refer to the sites $\curly{1,\ldots, N}$ as the \emph{bulk}, $B$.
We define the transverse-field Ising model (TFIM) Hamiltonian
\begin{equation}\label{eq:HExtended}
	\wtil{H} := \suml{k=0}{N}{a_{2k+1}\px{k}\px{k+1}} - \suml{k=1}{N}{a_{2k}\pz{k}}.
\end{equation}
on the extended chain that reduces to $H$ when the edge is initialized to state $\ket{++}$.
Similarly, we define $\wtil\uti := e^{-i\wtil\ha t_N}$. 
Note that the operator $\wtil H$ (and hence $\wtil U$) acts trivially on $\ket{++}_E$, so this edge state does not change through the course of the evolution.
(Our results also hold using the edge state $\ket{--}_E$, which
is equivalent to negating the sign of the longitudinal fields in~\eqref{eq:Hform}.)
We then prove that in the Heisenberg picture, Pauli matrices on site $k$ map to the corresponding Pauli on site $N+1-k$ for all sites $k$ in the chain.

First, we map to the doubled chain of Majorana fermionic operators by defining
\begin{equation}\label{eq:MajoranaDefinition}
	\gamma_{2k} \coloneqq \parstring{0,k-1}\cdot\px{k},\quad
	\gamma_{2k+1} \coloneqq \parstring{0,k-1}\cdot\py{k}
\end{equation}
at each site,
where we have used the notation $\parstring{a,b}:=\prod_{j=a}^{b}(-\pz{j})$ for the Jordan-Wigner parity string between sites $a$ and $b$.
The $\gamma_k$ are Hermitian and satisfy the Majorana anti-commutation relations $\curly{\gamma_j,\gamma_k}=2\delta_{jk}$.
We also see that $\pz{k} = -i\gamma_{2k}\gamma_{2k+1}$ and $\px{k}\px{k+1} = i\gamma_{2k+1}\gamma_{2k+2}$,
leading~\eqref{eq:HExtended} to take the form
\begin{equation}
 \label{eq:HformMajorana}
    \wtil{H} = i\suml{k=1}{2N+1}{a_k\gamma_{k}\gamma_{k+1}}\,.
  \end{equation}
The Majoranas $\gamma_0,\gamma_{2N+3}$ do not appear in the sum, since $a_0=a_{2N+2}=0$.

\begin{lemma}\label{lem:MajoranaAction}
	The operation $\wtil \uti$ acts on the Majorana operators as
	\begin{equation}
		\wtil\uti \gamma_k\wtil\uti\dg =
			\begin{dcases*}
				\gamma_k & if $k=0,2N+3$,\\
				(-1)^{k-1}\gamma_{2N+3-k} & otherwise.
			\end{dcases*}
	\end{equation}
\end{lemma}
\begin{proof}
	For the first case, $\wtil H$ has no overlap with operators $\gamma_0$ and $\gamma_{2N+3}$,
	so they are stationary under evolution by $\wtil H$.

	For the remaining cases,
	we make an analogy with the dynamics of the $y$ component of the spin operator, $S_y$, for a spin $s=N+\frac{1}{2}$ particle (as in, e.g.,~\cite{Albanese2004}).
        The Heisenberg evolution of $\gamma_k$ corresponds to the rotation of
	the $S_z$ eigenstate $\ket{s,k-s-1}$ of magnetization $k-s-1$.
	Observing that
	\begin{equation}
	\frac{i\pi}{4t_N}\bra{s,m}S_y\ket{s,m\pr}= a_{s+m+1}(\delta_{m\pr(m+1)}-\delta_{m(m\pr+1)})
	\end{equation}
	(with $\hbar=1$), we can express~\eqref{eq:HformMajorana}
	in the bilinear form $\wtil H=\frac{1}{2}\vec\gamma^\dagger A \vec\gamma$,
	for the vector $\vec \gamma \coloneqq \begin{bmatrix} \gamma_1 & \gamma_2 & \dots & \gamma_{2N+2}\end{bmatrix}$
	and the matrix $A \coloneqq -\pi/(2t_N)S_y$ expressed in the $S_z$ basis.
	Using the Majorana commutation relations,
	we have $\dot{\vec\gamma} = i[\wtil H, \vec\gamma] = 2iA\vec\gamma$,
	so $\vec\gamma (t) = e^{2iAt}\vec\gamma(0)$.
	The Heisenberg evolution of $\gamma_k$ under $\wtil H$ for time $t_N$ is exactly analogous to the (Schr\"{o}dinger) time evolution of the state $\ket{s,k-s-1}$ under $S_y$ for time $\pi$. A $\pi$-rotation under $S_y$ maps
	\begin{equation}
		\ket{s,-s+k-1}\mapsto (-1)^{k-1}\ket{s,s-k+1},
		\label{eq:piFlip}
	\end{equation}
	and correspondingly, $\gamma_k(t_N) = (-1)^{k-1}\gamma_{2N+3-k}$.
\end{proof}

Note that \cref{eq:piFlip} can easily be verified for a spin-$1/2$ particle. Similarly, a spin-$s$ particle may be viewed as a system of $2s$ spin-$\frac{1}{2}$ particles with maximal total spin. In this picture, a $\pi$-rotation under $S_y$ corresponds to independent $\pi$-rotations of each small spin. Since the state $\ket{s,k-s-1}$ is represented by a permutation-symmetric state with $k-1$ up spins, the $\pi$-rotation maps it to a state with $2s-(k-1)$ up spins and introduces a phase $(-1)$ for each up spin,
which is precisely~\eqref{eq:piFlip}.

Due to the signed reversal of the Majoranas in~\cref{lem:MajoranaAction}, the parity string $\parstring{0,k} = i^{b+1-a} \prod_{j=2a}^{2b+1} \gamma_j$ is (with the exception of $\gamma_0$) reflected about the center of the chain with an overall phase that exactly cancels when the product is reordered by increasing site index. The invariance of the edge Majoranas is crucial, as it provides a phase factor that cancels the state-dependent phases when we revert to the spin picture. In particular, we have the following lemma. 
\begin{lemma}\label{lem:StringAction}
	The operation $\wtil U$ acts on the parity strings as
$\wtil\uti\parstring{0,k}\wtil\uti\dg = i\px{0}\px{N+1}\parstring{0,N-k}$ for all $k$.
\end{lemma}
\begin{proof}
  Applying~\cref{lem:MajoranaAction}, we have
	\begin{align}
		\wtil\uti\parstring{0,k}\wtil\uti\dg
	  		&= i^{k+1} {(-1)}^{k(2k+1)} \gamma_0 \prod_{j=1}^{2k+1} \gamma_{2N+3-j}\,.\\
			&= \gamma_0 \parstring{0,N}\parstring{0,N-k}\gamma_{2N+2} \,
                   \end{align}
                   where we reordered the product and used $\parstring{N+1-k,N}=\parstring{0,N}\parstring{0,N-k}$. From the Majorana anti-commutation relations and~\eqref{eq:MajoranaDefinition}, the result follows.
\end{proof}

Now we prove the main theorem.
\begin{proof}[Proof of \cref{thm:main}]
  $U \cong \revfull$ holds iff all bulk observables on the chain transform identically under $U,\revfull$.
  For any operator $\mc O^{k}$ supported on bulk site $k\in\curly{1,\ldots, N}$, we show that $U\mc O^{k} U\dg = \bra{++}\wtil U\mc O^{k}\wtil U\dg\ket{++}_E = \mc O^{N+1-k}$. (Henceforth we drop the edge subscript $E$.)
	By~\cref{eq:MajoranaDefinition,lem:MajoranaAction,lem:StringAction}, $\sigma_x^k$ is mapped to
	\begin{align}
		\uti\px{k}\uti^\dagger
			&= \bra{++}\wtil\uti\parstring{0,k-1}\gamma_{2k}\wtil\uti^\dagger\ket{++} \\
			&= -i\bra{++}\px{0}\px{N+1}\parstring{0,N+1-k}\gamma_{2N+3-2k}\ket{++} \\
			&= -i\pz{N+1-k}\py{N+1-k} = \px{N+1-k}\,.
	\end{align}
	Next, we use \cref{lem:StringAction} to show that $\sigma_z^k$ is mapped to
	\begin{align}
		\uti \pz{k}  \uti^\dagger
			&=-\bra{++}\wtil\uti \parstring{0,k-1}\parstring{0,k} \wtil\uti^\dagger\ket{++} \\
			&= \bra{++}\px{0}\px{N+1}\parstring{0,N+1-k}\px{0}\px{N+1}\parstring{0,N-k}\ket{++}\\
			&= \pz{N+1-k}\,.
	\end{align}
	All other observables can be written in terms of the on-site Pauli operators $\px{k},\pz{k}$, so $U$ is identical to $\revfull$, up to global phase.
\end{proof}

\subsection{Time lower bound.}
We now prove a lower bound on the optimal time, $t^*$, to implement state reversal
using normalized local interactions.
Let the entanglement entropy between systems $A$ and $B$ of a bipartite state $\ket{\psi}_{AB}$
be $\entanglement{\ket{\psi}}$,
defined as the local von Neumann entropy $S(\rho) \coloneqq -\tr[\rho \log_2 \rho]$,
for $\rho = \tr_B[\ket\psi\bra\psi]$.
Then, the asymptotic entanglement capacity of a Hamiltonian $H$ that couples systems $A$ and $B$
was shown to equal~\cite{Bennett2003}
\begin{equation}\label{eq:EntanglementCapacityAnc}
	E_H = \sup_{\ket{\psi} \in \mathcal H_{AA'BB'}} \lim_{t \to 0} \frac{\entanglement*{e^{-iHt}\ket{\psi}} - \entanglement{\ket{\psi}}}{t}\,,
\end{equation}
where $\mathcal H_{AA'BB'}$ is the Hilbert space of the bipartite systems $A$ and $B$
with arbitrarily large ancilla spaces $A'$ and $B'$, respectively.
In particular, for a Hamiltonian of the form $\sigma_x \otimes \sigma_x$,
\cite{Duer2001,Childs2003}~showed that
\begin{equation}\label{eq:XXEntanglingRate}
	\alpha \coloneqq E_{\sigma_x \otimes \sigma_x} = 2\max_y \sqrt{y(1-y)} \log_2\frac{y}{1-y} \approx 1.912.
\end{equation}
This is tighter than the more general small incremental entangling bound $E_H \leq c \norm{H} \log_2 d = 2$
for the conjectured $c=2$~\cite{Bravyi2007} (best known $c=4$~\cite{Audenaert2014})
and where the smallest dimension of $A$ or $B$ gives $d=2$.
Since $E$ is invariant under local unitaries, a direct corollary is that $E_{\sigma_y \otimes \sigma_y}=E_{\sigma_z \otimes \sigma_z}=\alpha$.

We now show that \cref{prot:ti} is close to the shortest time possible.
\begin{theorem}\label{thm:lowerBound}
	It holds that $\frac{t_N}{t^*(1+1/N)} \leq \alpha \pi/4 < 1.502$.
\end{theorem}

\begin{proof}
We prove the time lower bound via an upper bound on the rate of increase of entanglement across a cut in the center of the chain (allowing differences of one qubit for odd $N$).
Designate the left half of the cut as subsystem $\mc A$ and the right half as subsystem $\mc B$.
$\mc A$ consists of subsystem $A$ given by the qubit at site $\floor{N/2}$ adjacent to the cut,
and subsystem $A'$ consisting of the remaining qubits to the left of the cut as well as a finite but arbitrary number of ancilla systems that are not part of the chain.
Similarly, $\mathcal B$ consists of subsystem $B$, the qubit at site $\floor{N/2}+1$,
and $B\pr$, the remaining qubits in the right half with an arbitrary finite number of ancilla.

Consider Hamiltonians of the form $H(t) = K(t) + \bar K(t)$ specifying the evolution of the $\mathcal A \mathcal B$ system,
where $K(t)$ is a two-qubit Hamiltonian supported on systems $AB$ (i.e., the cut edge),
while $\bar K$ contains terms supported on $AA\pr$ or $BB\pr$ but not the cut edge $AB$.
For brevity, we drop the time parameter $t$ even though we allow the Hamiltonian to be time-dependent.
We assume that $K$ is expressed in canonical form~\eqref{eq:CanonicalForm} due to equivalence under local unitaries.
Aside from its support, we make no assumptions about the form of $\bar K$
(so the resulting bound is more general than nearest-neighbor interactions).
We call $H$ satisfying these conditions \emph{divisible}
and also call protocols using divisible Hamiltonians divisible.

Observing that $E_H$ is the supremum over a time derivative of the von Neumann entropy of $\rho=\tr_{\mc B}\ket{\psi}\bra{\psi}$, we have
\begin{align}
  E_H &= \sup_{\ket \psi}\tr\paren{-\frac{d\rho}{dt}\log\rho - \rho\frac{d\log\rho}{dt}}\\
  	  &= \sup_{\ket \psi}\tr\paren{-\frac{d\rho}{dt}\log\rho}\,.\label{eq:EHderivative}
\end{align}
The reduced density matrix $\rho$ has time evolution
\begin{equation}
  \frac{d\rho}{dt} = -i\tr_{\mc B}\brac{H,\ket{\psi}\bra{\psi}}.
\end{equation}
We substitute $H=\bar{K}+\suml{j\in\curly{x,y,z}}{}{\mu_j\sigma_j\otimes\sigma_j}$ in the commutator and substitute the time-dependence of $\rho$ into~\cref{eq:EHderivative}.
By linearity of the trace and sublinearity of the supremum, we get
\begin{equation}
  	E_H \leq E_{\bar{K}} + \suml{j\in\curly{x,y,z}}{}{\mu_jE_{\sigma_j\otimes\sigma_j}} \leq \alpha\,,\label{eq:DivisibleRate}
\end{equation}
where we observe that $E_{\bar K}=0$ since $\bar{K}$ does not have support across the cut,
and use the normalization condition $\sum_j \abs{\mu_j} \leq 1$.
This bound holds for all divisible Hamiltonians $H$,
with nearest-neighbor Hamiltonians as a special case.

The entanglement generated by any divisble protocol can now be bounded in time.
We observe that if the protocol contains local measurements then these cannot increase entanglement $\entanglement*{\ket{\psi}}$
and that feedback may be viewed as a particular time-dependence of $H$ conditioned on measurement outcomes.
Therefore,
\eqref{eq:DivisibleRate} bounds the total increase in entanglement across bipartition $\mc {AB}$ over a time $t^*$ by
\begin{equation}
  \label{eq:DeltaEbound}
  \entanglement*{\ket{\psi(t^*)}} - \entanglement*{\ket{\psi(0)}} \leq \alpha t^*
\end{equation}
for any initial state $\ket{\psi(0)}$ acted on by a divisible protocol and LOCC.

Finally, we give an explicit bound on the worst-case time of divisible state reversal protocols by specifying an initial state.
Let the system start in the product state $\ket{\phi}_{\mathcal A} \otimes \ket{\phi}_{\mathcal B}$
where each qubit forms a Bell state with a local ancilla not part of the chain.
Clearly, $\entanglement{\ket{\phi}_{\mathcal A} \otimes \ket{\phi}_{\mathcal B}}=0$.
We perform a reversal $\revfull$ on the chain and get the state $\ket{\psi}_{\mathcal{AB}} \coloneqq \revfull (\ket{\phi}_{\mathcal A} \otimes \ket{\phi}_{\mathcal B})$,
which is maximally entangled, i.e., $\entanglement{\ket{\psi}_{\mathcal{AB}}} = N$.
Then, \eqref{eq:DeltaEbound} gives the bound
\begin{equation}
	t^* \geq \frac{\entanglement{\ket{\psi}_{\mathcal{AB}}} - \entanglement{\ket{\phi}_{\mathcal A} \otimes \ket{\phi}_{\mathcal B}}}{\alpha} \geq \frac{N}{\alpha}
\end{equation}
on any divisible state reversal protocol.
Comparing this to our protocol time~\eqref{eq:reversalTime}, we have
\begin{equation*}
	\frac{t_N}{t^*} \leq \frac{\alpha\pi \sqrt{(N+1)^2 - p(N)}}{4N} \leq \frac{\alpha\pi(1+1/N)}{4}\,.\qedhere
\end{equation*}
\end{proof}

\subsection{Discussion.}
\label{sec:discussion} The time-dependent protocol in \cite{Raussendorf2005} is closely related to our time-independent protocol, and both can be described within the same framework (see appendix). In the time-dependent case, the state is evolved alternately under two restrictions of the Hamiltonian~\eqref{eq:Hform}:
$H(\vec 1,\vec 0)$ (uniform Ising) and $H(\vec 0,\vec 1)$ (uniform transverse field),
each for time $\pi/4$, for a total of $N+1$ rounds.
In the Majorana picture, these Hamiltonians carry out a simultaneous braiding of neighboring Majoranas along even (resp.\ odd) edges of the doubled Majorana chain.
The resulting map matches \cref{lem:MajoranaAction} exactly, implying that the two protocols are identical at the level of Majorana operators. Indeed, any protocol achieving the map in \cref{lem:MajoranaAction} is guaranteed to implement state reversal.

In fact, as shown in the appendix, there is an infinite family of nearest-neighbor, time-independent Hamiltonian protocols for state reversal that generalizes~\cref{prot:ti}. The family is parameterized by a non-negative integer $m$, with modified $\px{k}\px{k+1}$ coupling $J_k^{(m)}\propto \sqrt{\paren{2N+1-2k+4m}\paren{2k+1+4m}}$ and unmodified $\pz{k}$ field strength. \cref{prot:ti} corresponds to the special case of $m=0$. By choosing large $m$, the coupling strength can be engineered to be nearly uniform throughout the chain, which may be a desirable feature in experimental implementations of the protocol~\cite{Karbach2005}. 

State reversal implements a specific permutation of qubits in a spin chain
faster than naively possible using \swap{}s.
More generally, we would like to know how we can perform qubit routing on a spin chain faster than possible naively (for example, by using fast state reversal as a subroutine).
Moreover, while our lower bound shows that sublinear scaling is not possible for qubit routing on a spin chain,
it is still an open question whether a superconstant advantage over routing using \swap{}s is possible with other interaction structures.
Answers to these questions have applications in circuit transformations for quantum architectures~\cite{Childs2019},
where qubit routing is a key subroutine.

\begin{acknowledgments}
	A.B. and A.G. acknowledge support by DoE ASCR Quantum Testbed Pathfinder program (award number DE-SC0019040), DoE ASCR FAR-QC (award number DE-SC0020312), DoE BES Materials and Chemical Sciences Research for Quantum Information Science program (award number DE-SC0019449), NSF PFCQC program, AFOSR, ARO MURI, ARL CDQI, and NSF PFC at JQI.
	E.S. and A.C. acknowledge support by the U.S. Department of Energy, Office of Science, Office of Advanced Scientific Computing Research, Quantum Testbed Pathfinder program (award number DE-SC0019040) and the U.S. Army Research Office (MURI award number W911NF-16-1-0349).
\end{acknowledgments}

\appendix
\section*{Appendix}
\subsection{Time-dependent protocol for reversal.}
\label{sec:appA}
\begin{figure}[tb]
 \centering
 \includegraphics[width=\linewidth]{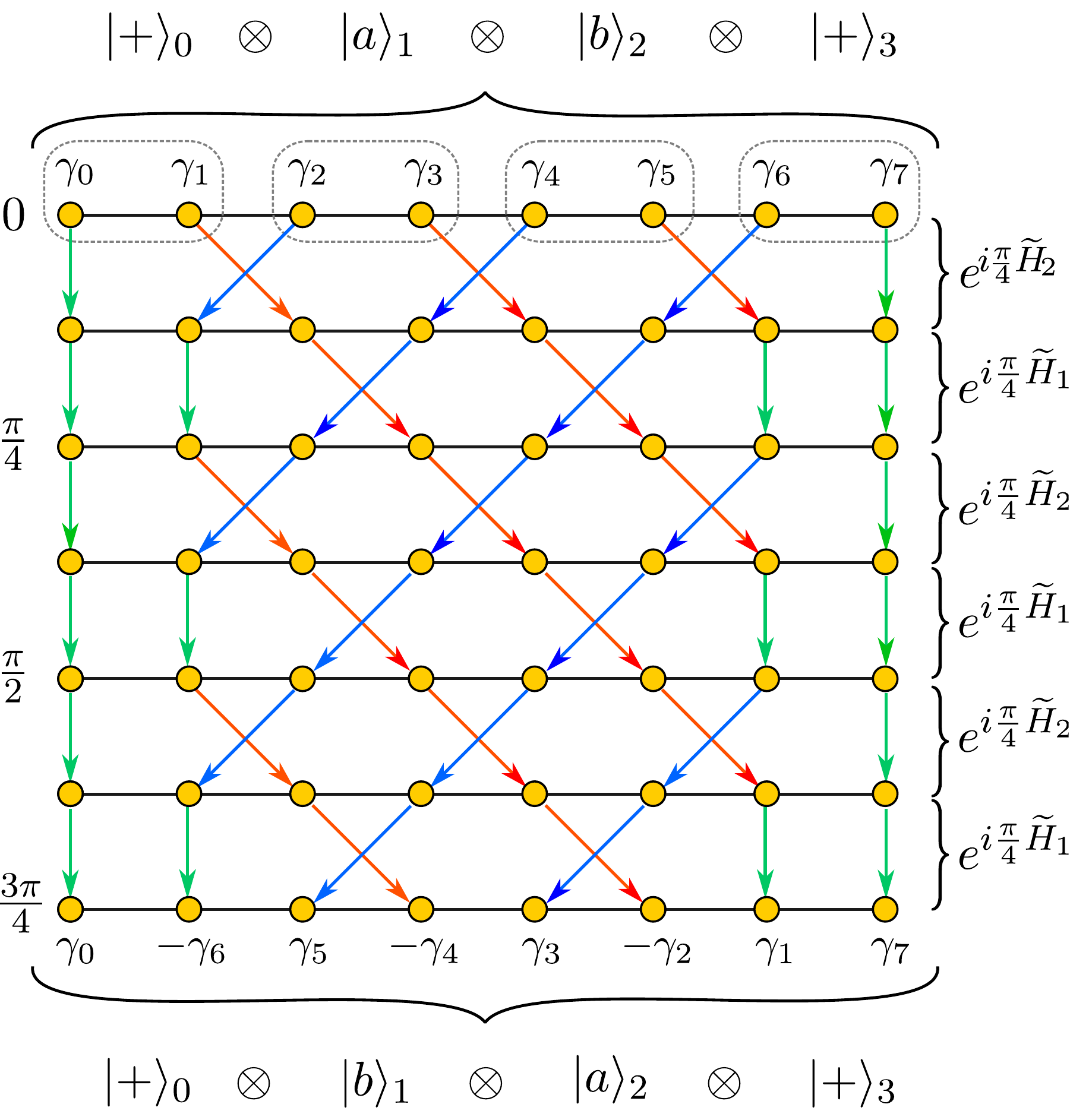}
 \caption{\small Time-dependent reversal protocol for $N=2$ (with two edge ancillas). For any bulk state $\ket{ab}_{12}$ (with edge state $\ket{++}_E$), alternating $\pi/4$ evolutions under $\wtil H_2$, $\wtil H_1$ are applied a total of $2N+2$ times. Each step braids neighboring Jordan-Wigner Majoranas; the right-movers (red) keep the same sign while the left-movers (blue) gain a minus sign. The edge Majoranas $\gamma_0,\gamma_7$ are unchanged (a crucial feature that ensures the correct parity phases), while the intermediate Majoranas undergo reversal of position with alternating sign. The final state in the bulk of the chain is $\ket{ba}_{12}$. }
 \label{fig:td}
\end{figure}
In this section, we give a simple analysis of the time-dependent protocol given in~\cite{Raussendorf2005, Fitzsimons2006} using our methods. The strategy is to prove that this protocol satisfies~\cref{lem:MajoranaAction} from the main text. \cref{lem:StringAction} and \cref{thm:main} are then automatically satisfied. First, we re-introduce the protocol using our notation.

\begin{protocol}\label{prot:td}
 Let $\hh := H(\vec 0,\vec 1)$ and  $\hj := H(\vec 1,\vec 0)$,
 where $\vec 1 = (1,1,\dots, 1)$ and $\vec 0 = (0,0,\dots,0)$. Explicitly, 
\begin{align}
  \hh &= \suml{k=1}{N}{Z_k}\,,\\
     \hj&= X_1 + \suml{k=1}{N-1}{X_kX_{k+1}} + X_N\,.
\end{align}
Apply $\utd:= \paren{e^{i\frac{\pi}{4}\hh}e^{i\frac{\pi}{4}\hj}}^{N+1}$ to the input state.
\end{protocol}
As in the main text, we extend the chain with two ancillary sites $\curly{0,N+1}$ that constitute the edge $E$. The unitary $\utd$ extends to an operator $\wtil\utd\coloneqq \idm_E\otimes\utd$ on the extended chain. Then the following lemma holds.

\begin{lemma}\label{lem:MajoranaActionTD}
	The operation $\wtil \utd$ acts on the Majorana operators as
	\begin{equation}
		\wtil\utd \gamma_k\wtil\utd\dg =
			\begin{dcases*}
				\gamma_k & if $k=0,2N+3$,\\
				(-1)^{k-1}\gamma_{2N+3-k} & otherwise.
			\end{dcases*}
	\end{equation}
\end{lemma}
\begin{proof}
  We use~\eqref{eq:MajoranaDefinition} to write $\utd$ as a product of alternating $\pi/4$-rotations
  under two Hamiltonians $\wtil\hj = i\suml{k=0}{N}{\gamma_{2k+1}\gamma_{2k+2}}$ and $\wtil\hh = i\suml{k=1}{N}{\gamma_{2k}\gamma_{2k+1}}$.
  Since $e^{-\pi/4 \gamma_i\gamma_j}$ is a braiding unitary that maps $\gamma_i\mapsto \gamma_j, \gamma_j\mapsto -\gamma_i, \gamma_{k\neq i,j}\mapsto \gamma_k$, it follows that the operator $e^{i\frac{\pi}{4}\wtil\hh}$ braids nearest-neighbor Majoranas along all odd edges of the chain (except the first and last edge), while $e^{i\frac{\pi}{4}\wtil\hj}$ braids along the even edges.
  Therefore, alternating $\pi/4$ rotations under $\wtil\hj$ and $\wtil\hh$ implement an even-odd sort \cite{Knuth1998} on the chain, as shown in \cref{fig:td}.
  Accounting for sign changes, the Majoranas map as follows: $\gamma_k\mapsto (-1)^{k+1}\gamma_{2N+3-k}$,
  while $\gamma_0,\gamma_{2N+3}$ remain unchanged.
\end{proof}

\subsection{Infinite family of Hamiltonians for state reversal.}
\label{sec:appB}
Reference \cite{Karbach2005} shows that there is an infinite family of XY Hamiltonians that generalize the protocol introduced in \cite{Albanese2004}.
In fact, \cref{prot:ti} is also a special case of an infinite family of protocols parameterized by a single non-negative integer $m$, as given below.
\begin{protocol}
  \label{prot:tim}
  Let $m\in\mathbb{Z}_{\ge 0}$, and
  \begin{align}
    J_k^{(m)}&\coloneqq \frac{\pi}{4}\sqrt{\paren{2k+1+4m}\paren{2N+1-2k+4m}} \\
    h_{k}^{(m)}&\coloneqq \pi\sqrt{k\paren{N+1-k}}
  \end{align}
for all sites $k=1,\ldots, N$. Let $\ha^{(m)} = H(\vec J^{(m)},\vec h^{(m)})$.
Apply $\uti^{(m)}:= e^{-i\ha^{(m)}}$ to the input state.
\end{protocol}

The protocol modifies only the couplings $J_{k}^{(m)}$ as a function of $m$, while the field terms $h_{k}^{(m)}=h_k$ are invariant with $m$. Note that $\uti^{(0)}=\uti$, so~\cref{prot:ti} is indeed a special case of~\cref{prot:tim}. For convenience, we have rescaled the coefficients so that the evolution time is $1$. To prove the correctness of this family of protocols, write the Hamiltonian $\ha^{(m)}$ in terms of Majorana fermions obtained by Jordan-Wigner transformation on the spin chain (extended to edge sites $\curly{0, N+1}$). We have
\begin{equation}
  \ha^{(m)} = \frac{1}{2}\vec{\gamma}\cdot A^{(m)}\cdot \vec{\gamma},
\end{equation}
where $\vec\gamma = \begin{bmatrix}\gamma_1 & \gamma_2 & \cdots & \gamma_{2N+2}\end{bmatrix}$ and $A^{(m)}$ is a $(2N+2)\times (2N+2)$ tridiagonal matrix with entries
\begin{multline}
  \label{eq:AmEntries}
  A^{(m)} =
  i\begin{pmatrix} 0 & J^{(m)}_{0} & & & & \\
    -J^{(m)}_{0} & 0 & h_1 & & & \\
      & -h_{1} & 0 & J^{(m)}_{1}& & \\
     & &\ddots & \ddots & \ddots & \\
     & & & -h_N & 0 & J_{N}^{(m)}\\
     & & & & -J_{N}^{(m)}& 0
 \end{pmatrix}.
\end{multline}
As before, the Heisenberg evolution of the Majoranas under $\ha^{(m)}$ is given by $\vec\gamma (t) = e^{2iA^{(m)}t}\vec \gamma(0)$. \cref{lem:MajoranaAction} shows that the operator $e^{2iA^{(0)}}$ implements reversal. Here we show that $e^{2iA^{(m)}} = e^{2iA^{(0)}}$ for all $m$, which implies that $\uti^{(m)}$ implements state reversal for all $m$. We state the following lemma (due to \cite{Oste2017, Hald1976}) on the spectrum of $A^{(m)}$.
\begin{lemma}\label{lem:tridiagonal}
  Let $A^{(m)}$ be as given in~\cref{eq:AmEntries}, and $s_k \coloneqq \mathrm{sgn}(2N+3-2k)$. Then $A^{(m)}$ has spectrum
  \begin{equation}
E_{k}^{(m)} = \frac{\pi}{4}\paren{2k-2N-3 + 4s_km}
\end{equation}
for $k=1,\ldots, 2N+2$. The corresponding eigenvectors $v_k$ satisfy $v_{kj} = (-1)^{N+k-j + 1/2}v_{k(2N+3-j)}$.
\end{lemma}
\begin{proof}
  The first claim follows from~\cite{Oste2017}. Via a transformation of the off-diagonals that preserves the spectrum, $A^{(m)}$ can be converted to a matrix $B(n,a)$ of Sylvester-Kac type
  \begin{equation}
    B(n,a)\coloneqq \frac{\pi}{4}
    \begin{pmatrix}
    0 & 1+a & & & & \\
    n+a & 0 & 2 & & & \\
    & n-1 & 0 & 3+a & & \\
    & & \ddots & \ddots & \ddots & \\
    & & & 2 & 0 & n+a \\
    & & & & 1+a & 0
    \end{pmatrix},
 \end{equation}
 for $n=2N+1, a=4m$. As shown in \cite{Oste2017}, the eigenvalues of $B(n,a)$ are given by the formula $\lambda_{\pm,j}=\pm \frac{\pi}{4}|2j+1+a|$ for $j\in\curly{0,\ldots, n}$, and the first claim follows.
 
  For the second claim, we observe again that $A^{(m)}$ may be converted to a real, symmetric, tridiagonal matrix $C^{(m)}$ with positive off-diagonal entries via the similarity transformation $C^{(m)}\coloneqq DA^{(m)}D^{-1}$ where $D = \mathrm{diag}\paren{i,i^2,\ldots, i^{2N+2}}$. Reference~\cite{Hald1976} shows that the eigenvectors $u_k = Dv_k$ of $C^{(m)}$ (ordered by ascending eigenvalue) satisfy $u_{kj} = (-1)^{k-1}u_{k(2N+3-j)}$ for $k=1,\ldots, 2N+2$. Correspondingly, the eigenvalues of $A^{(m)}$ satisfy $v_{kj} = (-1)^{k-1}i^{2N+3-2j}v_{k(2N+3-j)} = (-1)^{N+k-j + 1/2}v_{k(2N+3-j)}$.
\end{proof}
Finally, we show that $e^{2iA^{(m)}}$ implements reversal. 
\begin{theorem}
  \label{thm:IntegerSpectrum}
  For all $m\in\mathbb{Z}_{\ge 0}$, $A^{(m)}$ satisfies $\brac{e^{2iA^{(m)}}}_{jl} = (-1)^{j-1}\delta_{j(2N+3-l)}$.
\end{theorem}
\begin{proof}
  Write
  \begin{equation}
  e^{2iA^{(m)}} = \suml{k=1}{2N+2}{e^{2iE_k^{(m)}} v_k v_k\dg} = \suml{k=1}{2N+2}{(-1)^{k-N-3/2} v_{k} v_{k}\dg}\,,
\end{equation}
where we dropped the trivial phase $2\pi i m s_k$. The matrix elements of $e^{iA^{(m)}}$ are
  \begin{align}
    \brac{e^{iA^{(m)}}}_{jl} &= \suml{k=1}{2N+2}{(-1)^{k-N-3/2} v_{kj} v_{kl}^*}\\ &= \suml{k=1}{2N+2}{(-1)^{2N+2-l}v_{kj} v_{k(2N+3-l)}^*}\\  &= (-1)^{j-1}\delta_{j(2N+3-l)}\,,
  \end{align}
where in the second step we used \cref{lem:tridiagonal} as $v_{kl}^* = (-1)^{l-k-N-1/2}v_{k(2N+3-l)}^*$.
Therefore, $e^{2iA^{(m)}}$ maps $\gamma_k\mapsto (-1)^{k-1}\gamma_{2N+3-k}$, which implies that the protocol $\uti^{(m)}$ implements state reversal for all $m\in\mathbb{Z}_{\ge 0}$.
\end{proof}

When normalized so that all two-qubit terms are bounded by unity in spectral norm, $\ha^{(m)}$ implements state reversal in time $t_{N}^{(m)} = \frac{(N+1+4m)\pi}{4}$. Therefore, the time cost increases linearly in $m$ and is minimal for \cref{prot:ti} where $m=0$. Next, observe that if we choose $4m\gg N$, the variation in coupling coefficients $J^{(m)}_k$ is small and on the order ${\sim} \frac{1}{8}\paren{\frac{N+1}{2m}}^2$. Therefore, the parameter $m$ quantifies a trade-off between reversal time and the non-uniformity of $J_{k}^{(m)}$. Setting $m = N+1$, for example, yields a variation in the couplings on the order of $3\%$ for any $N$, and gives reversal in time $5N\pi/4$.

\bibliography{bibliography}

\end{document}